\documentclass[12pt]{article}

\usepackage{natbib}
\usepackage{natbib}
\usepackage{float}
\usepackage{microtype}
\usepackage[font=scriptsize]{caption}
\usepackage{amsthm}
\usepackage{appendix}
\usepackage{mathrsfs}
\usepackage{amssymb}
\usepackage{amsthm}
\usepackage{mathrsfs}
\usepackage{amssymb}
\usepackage{setspace}
\usepackage{amsmath}
\usepackage{tikz}
\usepackage{arydshln}
\usepackage{xcolor}
\usepackage{enumitem}
\setlist{nolistsep}

\colorlet{drkblue}{blue!61.8!black}

\usepackage[bookmarks=false]{hyperref}
\hypersetup{
 pdfborder = {0 0 0},
 colorlinks=true,
 citecolor=drkblue,
 linkcolor=drkblue,
 urlcolor=drkblue,
pdfstartview={XYZ null null 0.80}
}

\usepackage[left=2.5cm,top=3.5cm,bottom=3.5cm,right=3cm]{geometry}
\usepackage{setspace}
\theoremstyle{remark} \newtheorem{remark}{Remark}
\theoremstyle{plain} \newtheorem{lemma}{Lemma}
\theoremstyle{remark} 
\theoremstyle{plain} 
\newtheorem{example}{Example}
\newtheorem{theorem}{Theorem}
\newtheorem{corollary}{Corollary}
\theoremstyle{plain} \newtheorem{proposition}{Proposition}

\title{Lexicographic Choice Under Variable Capacity Constraints\thanks{Battal Do\u{g}an gratefully acknowledges financial support from  the British Academy/Leverhulme Trust through grant SR1819\textbackslash 190133 and the Swiss National Science Foundation (SNSF) through grant 100018\_ 162606. Kemal Y{\i}ld{\i}z gratefully acknowledges financial support from the Scientific and Research Council of Turkey (TUBITAK). 
We thank Mustafa O\u{g}uz Afacan, Samson Alva, Bettina Klaus, Bumin Yenmez, seminar participants at ITU Matching Workshop, Tel Aviv University, BRIC London 2017, and anonymous referees for valuable comments.}}
\author{Battal Do\u{g}an\thanks{Department of Economics, University of Bristol; battal.dogan@bristol.ac.uk.} \and Serhat Do\u{g}an \thanks{Department of Economics, Bilkent University; dserhat@bilkent.edu.tr.}\and Kemal Y{\i}ld{\i}z\thanks{Department of Economics, Bilkent University; kemal.yildiz@bilkent.edu.tr.}}
\date{\today}

\begin{document}

\maketitle

\begin{abstract}

In several matching markets, in order to achieve diversity, agents' priorities are allowed to vary across an institution's available seats, and the institution is let to choose agents in a lexicographic fashion based on a predetermined ordering of the seats, called a \textit{(capacity-constrained) lexicographic choice rule}. We provide a characterization of lexicographic choice rules and a characterization of deferred acceptance mechanisms that operate based on a lexicographic choice structure under variable capacity constraints.\@ We discuss some implications for the Boston school choice system and show that our analysis can be helpful in applications to select among plausible choice rules.\\

\noindent \textit{JEL} Classification Numbers: C78, D47, D78. 
\end{abstract}
\newpage

\newpage

\section{Introduction}

Many real-life resource allocation problems involve the allocation of an object that is available in a limited number of identical copies, called the \textit{capacity} of the object. Choice rules, which are systematic ways of rationing available copies of an object when demand exceeds the capacity, are essential in the analysis of such problems. A well-known example is the school choice problem in which each school has a certain number of seats to be allocated among students.\@ Although student preferences are elicited from the students, endowing each school with a choice rule is an essential part of the design process.  

Which choice rule to use is not always evident. The school choice literature, starting with the seminal study by \cite{as2003}, has widely focused on problems where each school is already endowed with a priority ordering over students and chooses the highest priority students up to the capacity. Such a choice rule, which is merely responsive to a given priority ordering, is called a \textit{responsive choice rule}.\footnote{In Appendix~\ref{responsive}, we discuss responsive choice rules.} However, when there are additional concerns such as achieving a diverse student body or affirmative action, which choice rule to use is non-trivial.\@ For example in the Boston school choice system, although each school is still endowed with a priority ordering over students and respecting student priorities is still a concern, schools would like to promote the neighborhood students as well by sometimes letting them override the priorities of students who are not from the neighborhood. Such an objective can obviously not be achieved with a responsive choice rule. 

The affirmative action policies that are in use in several school districts\footnote{In order to achieve a diverse student body, many school districts have been implementing affirmative action policies, such as in Boston, Chicago, and Jefferson County.} reveal that a natural way to achieve diversity is to allow students' priorities to vary across a school's seats, and to let the school choose students in a lexicographic fashion based on a predetermined ordering of the seats.\@ We call these rules \textit{(capacity-constrained) lexicographic choice rules}.\footnote{These choice rules are simply called \textit{lexicographic choice rules} in the recent market design literature. We introduce these choice rules using the \textit{capacity-constrained lexicographic choice} terminology to differentiate them from other lexicographic choice rules without capacity constraints which have been studied in the choice theory literature. Although we omit the ``capacity-constrained" part for simplicity in most part of the paper, we include it in the statements of our results.} To be more precise, a lexicographic choice rule specifies an ordering of the seats and assigns a priority ordering to each seat, which can be interpreted as the criterion based on which that particular seat will be assigned. At each choice set, the highest priority student according to the priority ordering at the first seat is chosen, then the highest priority ordering among the remaining students according to the priority ordering at the second seat is chosen, and so on until the last seat is assigned or no student is left.\@ Although some properties of lexicographic choice rules have already been studied in the literature, which set of properties distinguish lexicographic choice rules from other plausible choice rules has so far not been studied.\footnote{Although lexicographic choice rules are used to achieve diversity in school choice, there are other plausible choice rules that are also used, or can be used, to achieve diversity or affirmative action. Among others, \cite{ey2013} and \cite{ehyy2012} study some of those choice rules.} In this study, we follow the axiomatic approach and discover general principles (axioms) that characterize \textit{lexicographic choice rules} under \textit{variable capacity constraints}.\footnote{\cite{ey2013} also follow an axiomatic approach and characterize several choice rules for a school that wants to achieve diversity.}

In our baseline model, we consider a single decision maker who has a capacity constraint, such as a school with a limited number of seats. The decision maker encounters choice problems which consist of a choice set (a set of alternatives, such as students who demand a seat at the school) and a capacity. A choice rule, at each possible choice problem, chooses some alternatives from the choice set without exceeding the capacity. Note that across different choice problems, we allow capacity to vary, since in applications capacity may vary and the choice rule may need to be responsive to changes in capacity.\footnote{There are earlier studies in the literature which also formulate choice rules by allowing capacity to vary. See, among others, \cite{dogan_klaus_2016}, \cite{ehlers_klaus_2014a}, and \cite{ehlers_klaus_2014b}.} One example is when the number of available seats at a school may change from year to year. In fact, even during the same admissions year, a school may face two different choice problems with different capacities. In most of the existing school choice systems, such as New York City and Boston, there is a second stage of admissions including those students and school seats that are unassigned at the end of the first stage.\footnote{The new school choice system in Chicago also has two stages of admissions. See \cite{dogan_yenmez} for an analysis of the new system in Chicago.}

We consider the following three properties of choice rules that have already been studied in the axiomatic literature. 

\textit{Capacity-filling:\footnote{In the matching literature, \textit{capacity-filling} is also referred to as \textit{acceptance}, although the \textit{capacity-filling}
terminology has been increasingly popular in the recent literature.}} An alternative is rejected from a choice set at a capacity only if the capacity is full; 

\textit{Gross substitutes:} If an alternative is chosen from a choice set at a capacity, then it is also chosen from any subset of the choice set that contains the alternative, at the same capacity.

\textit{Monotonicity:} If an alternative is chosen from a choice set at a capacity, then it is also chosen from the same choice set at any higher capacity.

We introduce a new property called \textit{the irrelevance of accepted alternatives}.\@ \textit{The irrelevance of accepted alternatives} requires that, if the set of rejected alternatives is the same for two choice sets at the same capacity, then at any higher capacity, the set of accepted alternatives that were formerly rejected should be the same for the two choice sets.\@ In other words, in case of an increase in the capacity, \textit{the irrelevance of accepted alternatives} requires that the new alternatives that will be chosen (if any) should not depend on the already accepted alternatives. In Theorem~\ref{thm1}, we show that a choice rule satisfies \textit{capacity-filling}, \textit{gross substitutes}, \textit{monotonicity}, and \textit{the irrelevance of accepted alternatives} if and only if it is \textit{lexicographic}:  there exists a list of priority orderings over potential alternatives such that at each choice problem, the set of chosen alternatives is obtainable by choosing, first,  the highest ranked alternative according to the first priority ordering, then choosing the highest ranked alternative among the remaining alternatives according to the second priority ordering, and proceeding similarly until the capacity is full or no alternative is left.

Besides providing a first axiomatic foundation for lexicographic choice rules under variable capacity constraints, we also analyze the market design implications of lexicographic choice rules.\@ In Section~\ref{allocation_model}, we consider the variable-capacity object allocation model where there is more than one object (such as many schools) and agents have preferences over objects (such as students having preferences over schools).\@ In that model, \cite{ehlers_klaus_2014b} characterize deferred acceptance mechanisms where each object has a choice rule that satisfies \textit{capacity-filling}, \textit{gross substitutes}, and \textit{monotonicity}.\footnote{\cite{kojima_manea_2010} consider a setup where the capacity of each school is fixed, and characterize deferred acceptance mechanisms where each school has a choice rule that satisfies \textit{capacity-filling} and \textit{gross substitutes}.} Motivated by \textit{the irrelevance of accepted alternatives} for choice rules, we introduce a new property for allocation mechanisms, called \textit{the irrelevance of satisfied demand}. Consider an arbitrary problem and the allocation chosen by the mechanism at that problem. Suppose that the capacity of an object is increased.\@ Now, some of the agents who prefer that object to their assignments at the initial allocation may receive the object due to the capacity increase.\@ \textit{The irrelevance of satisfied demand} requires that the set of agents who receive the object due to the capacity increase does not depend on the set of agents who initially receive the object. We show that there is no mechanism which satisfies \textit{the irrelevance of satisfied demand} together with some other desirable properties studied in  \cite{ehlers_klaus_2014b} (Proposition \ref{impossibility}).\@ In particular, lexicographic deferred acceptance mechanisms, which are deferred acceptance mechanisms that operate based on a lexicographic choice structure, violate \textit{the irrelevance of satisfied demand}, which stands in contrast to lexicographic choice rules satisfying \textit{the irrelevance of accepted alternatives}. However, we show that a weaker version of \textit{the irrelevance of satisfied demand}\textendash which requires the same 
 at any problem where there is only one available object\textendash  
  characterizes  lexicographic deferred acceptance mechanisms together with the desirable properties studied in  \cite{ehlers_klaus_2014b} (Proposition \ref{ek_corollary}).  

\cite{ks2012} study lexicographic deferred acceptance mechanisms in a more general matching with contracts framework \citep{hm2005}.\@ In some applications, the choice rule of an institution is subject to  \textit{a feasibility constraint}, in the sense that some alternatives cannot be chosen together with some other alternatives. The matching with contracts model due to \cite{hm2005} introduced a general framework that incorporates such feasibility constraints into the matching problem. Although for the school choice application, where such feasibility constraints are not binding, the lexicographic choice rules in \cite{ks2012} fall into our baseline model, in case of binding feasibility constraints, their lexicographic choice rules are not covered in our baseline analysis.\footnote{For instance, the lexicographic choice rules in their setup may violate ``substitutability'', which is a generalization of gross substitutes to the matching with contracts setup \citep{hm2005}.} In Section \ref{feasibility}, we show that our baseline model and our baseline properties can be extended to a setup with feasibility constraints, highlighting the distinguishing properties of capacity-constrained lexicographic choice rules, including the ones discussed in \cite{ks2012}, in a more general setup. 

Boston school district is one of the school districts that uses 
capacity-constrained lexicographic choice to achieve a diverse student body and implement affirmative action policies.
Boston school district aims to give priority to neighborhood applicants for half of each school's seats.\@ To achieve this goal, the Boston school district has been using a deferred acceptance mechanism based on a choice structure, where each school is endowed with a ``capacity-wise lexicographic'' choice rule, that is, at each capacity, the choice rule lexicographically operates based on a list containing as many priority orderings as the capacity, yet the lists for different capacity levels do not have to be related in any way.\footnote{See \cite{dur2016published} for a detailed discussion of Boston's school choice mechanism.} 
\cite{dur2016} and \cite{dur2016published} analyse how the order of the priority orderings in the choice rule of a school may cause additional bias for or against the neighbourhood students.\footnote{\cite{dur2016} is an earlier version of \cite{dur2016published}.} In Section~\ref{boston}, we consider a class of capacity-wise lexicographic choice rules discussed in \cite{dur2016} that are relevant for the design of the Boston school choice system and show that our analysis enables us to single out one rule from four plausible candidates.

The paper is organized as follows. In Section \ref{literature}, we review the related literature. In Section~\ref{capacity_constrained_choice}, we introduce and characterize lexicographic choice rules, show that our baseline model and our baseline properties can be extended to a setup with feasibility constraints, and also provide a characterization of responsive choice rules. In Section~\ref{allocation_model}, we highlight an implication of our choice theoretical analysis for the resource allocation framework: we provide a  characterization of deferred acceptance mechanisms that operate based on a lexicographic choice structure. In Section~\ref{boston}, we discuss some implications for the Boston school choice system. In Section~\ref{discussion}, we conclude by discussing the main features of our analysis. 

\section{Related Literature}
\label{literature}

Several studies investigate choice rules that satisfy \textit{path independence} \citep{plott}, which requires that if the choice set is ``split up'' into smaller sets, and if the choices from the smaller sets are collected and a choice is made from the collection, the final result should be the same as the choice from the original choice set.
Since \textit{capacity-filling} together with \textit{gross substitutes} imply \textit{path independence},\footnote{This is also noted in Remark 1 of \cite{dogan_klaus_2016}, and it follows from Lemma 1 of \cite{ehlers_klaus_2014b} together with Corollary 2 of \cite{AizermanMalishevski81}.} lexicographic choice rules are examples of path independent choice rules. \cite{AizermanMalishevski81} show that for each path independent choice rule, there exists a list of priority orderings such that the choice from each choice set is the union of the highest priority alternatives in the priority orderings.\footnote{In the words of \cite{AizermanMalishevski81}, each \textit{path independent} choice rule is generated by some mechanism of collected extremal choice.}
Among others, \cite{plott}, \cite{moulin_1985}, and \cite{jd_2001} study the structure of path independent choice rules. Path independent choice rules guarantee the existence of stable matchings in the matching context.\@ \cite{chambers_yenmez_2017} study path independence in the matching context and its connection to stable matchings.

Although the structure of path independent choice rules have been extensively studied, the structure of lexicographic choice rules and what properties distinguish them from other path independent choice rules have not been well-understood. \cite{houy_tadenuma} consider two classes of choice rules which are both based on ``lexicographic procedures'', yet different than the ones we consider here. Similar to our setup, choice rules that they consider operate based on a list of binary relations.\footnote{\cite{houy_tadenuma} do not start with any assumptions on the list of binary relations. They separately discuss under which assumptions on the list of binary relations, the resulting choice rules satisfy certain properties.} Yet, their model does not include capacity constraints and the lexicographic procedures that operationalize the lists are different. The only study that considers lexicographic choice rules that we study from an axiomatic perspective is \cite{chambers_yenmez_2018b}. They show that lexicographic choice rules satisfy \textit{capacity-filling} and \textit{path independence}, and they also show that there are \textit{path independent} choice rules that are not lexicographic, but they do not provide a characterization of lexicographic choice rules.

Our analysis of the Boston school choice system is related to \cite{dur2016published} and the working paper version \cite{dur2016}. \cite{dur2016} compare alternative choice rules for schools in the Boston school district (one of which is the one used in the Boston school district)  in terms of how much they are biased for or against the neighbourhood students. We consider these alternative choice rules from a different perspective. In Section \ref{boston}, we show that, although these choice rules are all based on a ``lexicographic procedure" at each capacity, only one of them satisfies all the characterizing properties in Theorem \ref{thm1}, and therefore only one of them is actually a \textit{(capacity-constrained) lexicographic choice rule}. The common feature of \cite{dur2016published} and our Section \ref{boston} is that we both consider lexicographic choice procedures in the context of school choice in Boston.\@ The main difference is that, although the choice rules that \cite{dur2016published} consider have direct counterparts in a variable capacity context, their analysis pertains to the fixed capacity case. In particular, given a fixed school capacity, \cite{dur2016published} analyze how different lexicographic choice procedures perform. On the other hand, variable capacities, and properties related to variable capacities, are at the heart of our study. We show that, one of our variable capacity properties, \textit{CWARP}, is satisfied by only one of the four choice rules discussed in \cite{dur2016}.

\section{Capacity-Constrained Lexicographic Choice}
\label{capacity_constrained_choice}

Let $A$ be a nonempty finite set of $n$ alternatives and let $\mathcal{A}$ denote the set of all \textit{nonempty} subsets of $A$. A (capacity-constrained) choice \textbf{problem} is a pair $(S,q)\in \mathcal{A} \times \{1,\ldots ,n\}$ of a choice set $S$ and a capacity $q$. A (capacity-constrained) \textbf{choice rule} $C: \mathcal{A} \times \{1,\ldots ,n\} \rightarrow \mathcal{A}$ associates with each problem $(S,q)\in \mathcal{A} \times \{1,\ldots ,n\}$, a set of choices $C(S,q)\subseteq S$ such that $|C(S,q)|\leq q$. Given a choice rule $C$, we denote the set of rejected alternatives at a problem $(S,q)$ by $R(S,q)=S\setminus C(S,q)$.\medskip

A \textbf{priority ordering} $\succ$ is a complete, transitive, and anti-symmetric binary relation over $A$. A \textbf{priority profile} $\pi=(\succ_1, \ldots ,\succ_n)$ is an ordered list of $n$ priority orderings. Let $\Pi$ denote the set of all priority profiles.

A choice rule $C$ is \textbf{(capacity-constrained) lexicographic for a priority profile} $(\succ_1, \ldots ,\succ_n)\in \Pi$ if for each $(S,q)\in \mathcal{A} \times \{1,\ldots ,n\}$, $C(S,q)$ is obtained by choosing the highest $\succ_1$-priority alternative in $S$, then choosing the highest $\succ_2$-priority alternative among the remaining alternatives, and so on until $q$ alternatives are chosen or no alternative is left. A choice rule is \textbf{(capacity-constrained) lexicographic} if there exists a priority profile for which the choice rule is lexicographic.\medskip

\begin{remark}
Note that, if a choice rule is lexicographic for a priority profile $\pi=(\succ_1, \ldots ,\succ_n)$, then it is lexicographic for any other priority profile that is obtained from $\pi$ by replacing $\succ_n$ with an arbitrary priority ordering. In that sense, the last priority ordering is redundant. 
\end{remark}

We consider four properties of choice rules. The following three properties are already known in the literature. 

\noindent \textbf{Capacity-filling:} An alternative is rejected from a choice set at a capacity only if the capacity is full. Formally, for each $(S,q)\in \mathcal{A} \times \{1,\ldots ,n\}$, $$|C(S,q)|=\min\{|S|,q\}.$$

\noindent \textbf{Gross substitutes:}\footnote{\textit{Gross substitutes} was first introduced in the choice literature by \cite{chernoff}. It has been studied in the choice literature under different names such as \textit{Chernoff's axiom}, \textit{Sen's $\alpha$}, or \textit{contraction consistency}. In the matching literature, it was first studied and referred to as \textit{gross substitutes} in \cite{kelso_crawford_1982} (\textit{substitutability} is also a commonly used name in the matching literature). We follow the terminology of \cite{kelso_crawford_1982}.} If an alternative is chosen from a choice set at a capacity, then it is also chosen from any subset of the choice set that contains the alternative, at the same capacity. Formally, for each $(S,q)\in \mathcal{A} \times \{1,\ldots ,n\}$ and each pair $a,b\in S$ such that $a\neq b$, $$\mbox{if }a\in C(S,q),\mbox{ then }a\in C(S\backslash \{b\},q).$$

\noindent \textbf{Monotonicity:} If an alternative is chosen from a choice set at a capacity, then it is also chosen from the same choice set at any higher capacity. Formally, for each $(S,q)\in \mathcal{A} \times \{1,\ldots ,n-1\}$, $$C(S,q)\subseteq C(S,q+1).$$

We now introduce a new property called \textit{the irrelevance of accepted alternatives}. Consider a problem and the set of rejected alternatives for that problem. Suppose that the capacity increases. The property requires that which alternatives among the currently rejected alternatives will be chosen (if any) should not depend on the currently accepted alternatives. In other words, if the set of rejected alternatives are the same for two choice sets (note that the set of accepted alternatives may be different), then at any higher capacity, the set of initially rejected alternatives that become accepted should be the same for the two choice sets.  

\noindent \textbf{Irrelevance of accepted alternatives:} For each $S ,S' \in \mathcal{A}$ and each $q\in \{1,\ldots ,n-1\}$, $$\mbox{if }R(S,q)=R(S',q),\mbox{ then }C(S,q+1)\cap R(S,q) =C(S',q+1)\cap R(S',q).$$ 

We also introduce another property called \textit{capacity-wise weak axiom of revealed preference} which will be helpful in our analysis. Consider the following capacity-wise revealed preference relation. An alternative $a\in A$ is revealed to be preferred to an alternative $b\in A$ at a capacity $q>1$ if there is a problem with capacity $q-1$ for which $a$ and $b$ are both rejected and $a$ is chosen over $b$ when the capacity is $q$. That is, $a$ is \textbf{revealed to be preferred} to $b$ at $q$ if there exists $S\in \mathcal{A}$ such that $a,b\notin C(S,q-1)$, $a\in C(S,q)$, and $b\in R(S,q)$. \textit{Capacity-wise weak axiom of revealed preference} requires, for each capacity, the revealed preference relation to be asymmetric.

\noindent \textbf{Capacity-wise weak axiom of revealed preference (CWARP):} For each capacity $q>1$ and each pair $a,b\in A$, if $a$ is revealed to be preferred to $b$ at $q$, then $b$ is not revealed to be preferred to $a$ at $q$. 

\textit{CWARP} is a counterpart of the well-known \textit{weak axiom of revealed preference} (WARP) in the standard revealed preference framework \citep{samuelson}, where there is no capacity parameter.\@ In the standard framework, an alternative is said to be revealed preferred to another alternative if there is a choice set at which the former alternative is chosen over the latter.\@ \textit{WARP} requires the revealed preference relation to be asymmetric, which in a sense requires consistency of the choice behavior in responding to changes in the choice set. In our framework, the preference is revealed not only through the choice at a choice set, but also through a change in the capacity.\@ Therefore, what should be the counterpart of the ``revealed preference relation''  is not entirely clear.\@ We propose the following definition. An alternative is revealed to be preferred to another at a capacity if there is a choice set in which the former alternative is chosen over the latter at that capacity, although if the capacity were one less, none of the alternatives would have been chosen.\@ Put differently, if none of the two alternatives are chosen in a choice set at a given capacity, but one of them is chosen when capacity increases by one, this means the chosen alternative is revealed to be preferred to the unchosen one.\@   \textit{CWARP} requires the revealed preference relation to be asymmetric.\@ Hence, \textit{CWARP} requires consistency of the choice behavior in responding to changes in the choice set together with changes in the capacity. Additionally, one can interpret \textit{CWARP} as a  ``no complementarities'' condition, in the sense that \textit{CWARP}  requires the new alternative to be chosen due to the capacity increase be independent of the alternatives that have already been chosen.\@ For example, if two alternatives are complements, then the choice of each one of these alternatives may depend on whether the other one has already been chosen or not.\@ \textit{CWARP} rules out this type of choice behavior.

\begin{remark} The following is an alternative definition of \textit{CWARP}, which is formulated in line with the common formulations of WARP-type revealed preference relations in the literature.

\noindent\textit{An alternative definition of CWARP:} For each capacity $q>1$, each pair $S,T\in \mathcal{A}$ and each pair $a,b\in S\cap T$ such that $[C(S,q-1) \cup C(T,q-1)] \cap \{a,b\}= \emptyset$, 
$$\text{if } a \in C(S, q) \text{ and } b \in C(T,q)\setminus C(S,q), \text{ then } a \in C(T, q).$$
\end{remark}

\begin{lemma}
If a choice rule satisfies \textit{capacity-filling}, \textit{monotonicity}, and \textit{CWARP}, then it also satisfies \textit{the irrelevance of accepted alternatives}.
\label{lemma_rejection_monotonicity}
\end{lemma}

\begin{proof}
Let $C$ be a choice rule. Suppose that $C$ satisfies \textit{capacity-filling} and \textit{monotonicity}, but violates \textit{the irrelevance of accepted alternatives}. By violation of \textit{the irrelevance of accepted alternatives}, there are $S ,S' \in \mathcal{A}$ and $q\in \{1,\ldots ,n-1\}$ such that $R(S,q)=R(S',q)$, but $ C(S,q+1)\cap R(S,q)\neq C(S',q+1)\cap R(S',q) .$ By \textit{monotonicity}, $R(S,q+1)\subseteq R(S,q)$ and $R(S',q+1)\subseteq R(S',q)$. By \textit{capacity-filling}, $|R(S,q+1)|=|R(S',q+1)|.$ Then, there exist $a,b\in R(S,q)=R(S',q)$ such that $a\in C(S,q+1)$, $b\notin C(S,q+1)$, $b\in C(S',q+1)$, and $a\notin C(S',q+1)$. But then, $a$ is revealed preferred to $b$ and vice versa, implying that $C$ violates \textit{CWARP}.
\end{proof}

In Appendix \ref{necessity_lemma_1}, we show that each of the three properties \textit{capacity-filling}, \textit{monotonicity}, and \textit{CWARP} is necessary for the implication in Lemma \ref{lemma_rejection_monotonicity}, that is, we provide examples of choice rules which violate exactly one of the three properties and also violate \textit{the irrelevance of accepted alternatives}.

The following example shows that there exists a choice rule that satisfies \textit{capacity-filling}, \textit{monotonicity}, and \textit{the irrelevance of accepted alternatives}, but violates \textit{CWARP}.

\begin{example}
Let $A=\{a,b,c,d,e\}$. Let $\succ$ and $\succ'$ be defined as $a\succ b \succ c \succ d \succ e$ and $a\succ' c \succ' b \succ' d \succ' e$. Let the choice rule $C$ be defined as follows. For each problem $(S, q)$, if $d\in S$, then $C(S,q)$ chooses the highest $\succ$-priority alternatives from $S$ until $q$ alternatives are chosen or no alternative is left;\footnote{That is, $C(S,q)$ coincides with the choice rule that is ``responsive'' for $\succ$. We discuss responsive choice rules in Section~\ref{responsive}.} if $d\notin S$, then $C(S,q)$ chooses the highest $\succ'$-priority alternatives from $S$ until $q$ alternatives are chosen or no alternative is left. Note that $C$ clearly satisfies \textit{capacity-filling} and \textit{monotonicity}. To see that $C$ also satisfies \textit{the irrelevance of accepted alternatives}, let $S ,S' \in \mathcal{A}$ and $q\in \{1,\ldots ,n-1\}$ be such that $R(S,q)=R(S',q)$. If $d\in S\cap S'$ or $d\in A\setminus (S\cup S')$, then $ C(S,q+1)\cap R(S,q)= C(S',q+1)\cap R(S',q)$. So suppose, without loss of generality, that $d\in S\setminus S'$. Since $R(S,q)=R(S',q)$, we have $d\in C(S,q)$. But then, either $R(S,q)=\emptyset$ or $R(S,q)=\{e\}$. In either case, we have $ C(S,q+1)\cap R(S,q)= C(S',q+1)\cap R(S',q)$. To see that $C$ violates \textit{CWARP}, note that $C(\{a,b,c,d\},1)=\{a\}$ and $C(\{a,b,c,d\},2)=\{a,b\}$, implying that $b$ is revealed preferred to $c$ at $q=2$. Also, $C(\{a,b,c,e\},1)=\{a\}$ and $C(\{a,b,c,e\},2)=\{a,c\}$, implying that $c$ is revealed preferred to $b$ at $q=2$.  
\end{example}

\begin{theorem}
A choice rule is (capacity-constrained) lexicographic if and only if it satisfies \textit{capacity-filling}, \textit{gross substitutes}, \textit{monotonicity}, and \textit{the irrelevance of accepted alternatives}.\footnote{Independence of the characterizing properties is shown in Appendix~\ref{appendix:independence}.} 
\label{thm1}
\end{theorem} 

\begin{proof}
Let $C$ be lexicographic for $(\succ_1, \ldots ,\succ_n)\in \Pi$. Clearly, $C$ satisfies \textit{capacity-filling} and \textit{monotonicity}, and it is already known from the literature that $C$ satisfies \textit{gross substitutes} \citep{chambers_yenmez_2018b}. To see that it satisfies \textit{CWARP}, let $a,b\in A$ and $q\in \{2,\ldots ,n\}$ be such that $a$ is revealed preferred to $b$ at $q$. Then, there is $S\in \mathcal{A}$ such that $a,b\in R(S,q-1)$, $a\in C(S,q)$, and $b\in R(S,q)$. But then, $a\succ_q b$. If also $b$ is revealed preferred to $a$ at $q$, then by similar arguments we have $b\succ_q a$, contradicting that $\succ_q$ is antisymmetric. Thus, the revealed preference relation is asymmetric and $C$ satisfies \textit{CWARP}. By Lemma \ref{lemma_rejection_monotonicity}, $C$ also satisfies \textit{the irrelevance of accepted alternatives}.

Let $C$ be a choice rule satisfying \textit{capacity-filling}, \textit{gross substitutes}, \textit{monotonicity}, and \textit{the irrelevance of accepted alternatives}. We first construct a priority profile $(\succ_1, \ldots ,\succ_n)\in \Pi$ and then show that $C$ is lexicographic for that priority profile. For each $i,j\in \{1,\ldots ,n\}$, let $a_{ij}$ denote the $j$'th ranked alternative in $\succ_i$ (for instance, $a_{i1}$ is the highest $\succ_i$-priority alternative).

To construct $\succ_1$, first set $\{a_{11}\}=C(A,1)$. For each $j\in \{2,\ldots ,n\}$, set $\{a_{1j}\}=C(A\setminus \{a_{11},\ldots ,a_{1(j-1)}\},1)$.
To construct $\succ_2$, consider $C(A,2)$. By \textit{capacity-filling}, $|C(A,2)|=2$. Since $a_{11}\in C(A,1)$, by \textit{monotonicity}, $a_{11}\in C(A,2)$. Set $\{a_{21}\}=C(A,2)\setminus \{a_{11}\}$. For each $j\in \{2,\ldots ,n-1\}$, set $\{a_{2j}\}=C(A\setminus \{a_{21},a_{22},\ldots ,a_{2(j-1)}\},2)\setminus \{a_{11}\}$. Set $a_{2n}=a_{11}$. 

The rest of the priority profile is constructed recursively as follows. For each $i\in \{3,\ldots ,n\}$, first set $\{a_{i1}\}=C(A,i)\setminus \{a_{11},a_{21},\ldots ,a_{(i-1)1}\}$ (Note that by \textit{monotonicity}, $\{a_{11},a_{21},\ldots ,a_{(i-1)1}\}\subseteq C(A,i)$ and by \textit{capacity-filling}, $|C(A,i)|=i$). For each $j\in \{2,\ldots ,n-i+1\}$, set $\{a_{ij}\}=C(A\setminus \{a_{i1},a_{i2},\ldots ,a_{i(j-1)}\},i)\setminus \{a_{11},a_{21},\ldots ,a_{(i-1)1}\}$. Note that there are $i-1$ rankings yet to be set in $\succ_i$, which are $\{a_{i(n-i+2)},\ldots ,a_{in}\}$. For each $j\in \{n-i+2,\ldots ,n\}$, set $a_{ij}=a_{(j+i-n-1)1}$ (which assigns the alternatives $a_{11}, \ldots , a_{(i-1)1}$ to the rankings $a_{i(n-i+2)},\ldots ,a_{in}$, respectively). 

Now, let $(S,q)\in \mathcal{A} \times \{1,\ldots ,n\}$. Let $b_1$ denote the highest $\succ_1$-priority alternative in $S$, $b_2$ denote the highest $\succ_2$-priority alternative among the remaining alternatives, and so on up to $b_{\min\{|S|,q\}}$. We show that $C(S,q)=\{b_1,\ldots ,b_{\min\{|S|,q\}}\}$. If $\min\{|S|,q\}=|S|$, then by \textit{capacity-filling}, $C(S,q)=\{b_1,\ldots ,b_{|S|}\}$. Suppose that $|S|>q$.

The rest of the proof is by induction: we first show that $b_1 \in C(S,q)$; then, for an arbitrary $i\in \{2,\ldots ,q\}$, assuming that $b_1,\ldots ,b_{i-1}\in C(S,q)$, we show that $b_i\in C(S,q)$. Let $b_1=a_{1j}$ for some $j\in \{1,\ldots ,n\}$. By the construction of $\succ_1$, $b_1\in C(A\setminus \{a_{11},\ldots ,a_{1(j-1)}\},1)$. Then, by \textit{gross substitutes} and \textit{monotonicity}, $b_1 \in C(S,q)$.

Let $i\in \{2,\ldots ,q\}$. Assuming that $b_1,\ldots ,b_{i-1}\in C(S,q)$, we show that $b_i \in C(S,q)$. Let $S'$ be the choice set obtained from $S$ by replacing $b_1$ with $a_{11}$ (note that nothing changes if $b_1=a_{11}$), replacing $b_2$ with $a_{21}$, $\ldots$, and replacing $b_{i-1}$ with $a_{(i-1)1}$. That is, $S'=(S\setminus \{b_1,\ldots ,b_{i-1}\})\cup \{a_{11},\ldots ,a_{(i-1)1}\}$. Let $q'=i-1$. Note that $\{b_1,\ldots ,b_{i-1}\}=C(S,q')$, because otherwise, by \textit{capacity-filling}, there is $a\in S$ such that $a\in C(S,q')$ and $a\notin C(S,q)$, which is a violation of \textit{monotonicity}. Also, by the construction of the priority profile and by \textit{gross substitutes}, $\{a_{11},\ldots ,a_{(i-1)1}\}=C(S',q')$. Note that $R(S,q')=R(S',q')$. By \textit{monotonicity} and \textit{the irrelevance of accepted alternatives}, we have $R(S,q)=R(S',q)$. Since $b_i\in C(S',q)$ by the construction of the priority profile and by \textit{gross substitutes}, we also have $b_i\in C(S,q)$.  
\end{proof}

\begin{corollary}
A choice rule is (capacity-constrained) lexicographic if and only if it satisfies \textit{capacity-filling}, \textit{gross substitutes}, \textit{monotonicity}, and \textit{CWARP}. 
\label{corollary_rejection_monotonicity}
\end{corollary} 

\begin{proof}
A lexicographic choice rule satisfies \textit{capacity-filling}, \textit{gross substitutes}, and \textit{monotonicity} by Theorem~\ref{thm1}. Also note that in the proof Theorem~\ref{thm1}, we showed that a lexicographic choice rule satisfies \textit{CWARP} as well. To see the other direction, note that by Lemma~\ref{lemma_rejection_monotonicity}, \textit{capacity-filling}, \textit{monotonicity}, and \textit{CWARP} imply \textit{the irrelevance of accepted alternatives} and the rest follows by Theorem~\ref{thm1}.
\end{proof}

There is never a unique priority profile for which a given choice rule is lexicographic. However, if $C$ is lexicographic for two different priority profiles $(\succ_1, \ldots ,\succ_n)$ and $(\succ'_1, \ldots ,\succ'_n)$, then for each pair of alternatives $a,b\in A$, 
if $a\mathrel{\succ_t}b$ and $b\mathrel{\succ'_t}a$
for some $t\in \{1,\ldots ,n\}$, then $a$ or $b$ must be chosen from any choice set (particularly from $A$) at any capacity $q<t$. That is, $a$ or $b$  is chosen irrespective of their relative ranking at the $t$'th priority ordering.  

To state this observation formally, for each priority ordering $\succ_i$ on $A$ and for each choice set $S\in \mathcal{A}$, let $\succ_i|_S$ stand for the restriction of $\succ_i$ to $S$. 
Let $A_1=A$, and  for each $t\in \{2,\ldots ,n\}$, let $A_t=A\setminus C(A,t-1)$. For each choice set $S\in \mathcal{A}$ and each priority ordering $\succ_i$, let $max(S, \succ_i) $ be the top-ranked alternative in $S$ according to $\succ_i$.

\begin{proposition}
If a choice rule $C$ is (capacity-constrained) lexicographic for a priority profile $(\succ_1, \ldots ,\succ_n)$, then  $C$ is lexicographic for another priority profile $(\succ'_1, \ldots ,\succ'_n)$ if and only if $\succ_1\: = \: \succ'_1$ and for each $t\in \{1,\ldots ,n\}$, $\succ_t|_{A_t}\: =\:\succ'_t|_{A_t}$.
\end{proposition}

\begin{proof}

\textit{(If part)} Let choice rule $C$ be  lexicographic for a priority profile $(\succ_1, \ldots ,\succ_n)$. Suppose $(\succ'_1, \ldots ,\succ'_n)$ is such that $\succ_1\: = \: \succ'_1$ and for each $t\in \{1,\ldots ,n\}$, $\succ_t|_{A_t}\: =\:\succ'_t|_{A_t}$. Now, for each $S\in \mathcal{A}$ and $t\in \{1,\ldots ,n\}$, if $t=1$, then since $\succ_1=\succ'_1$, the conclusion is immediate. Then, by proceeding inductively, for each $1<t\leq|S|$,  since   $C$ is lexicographic for  $(\succ_1, \ldots ,\succ_n)$,  $max(S\setminus C(S,t-1), \succ_t) = C(S,t)\setminus C(S,t-1)$. Since $S\setminus C(S,t-1)\subset A_t$ and $\succ_t|_{A_t}\: =\: \succ'_t|_{A_t}$, we get $max(S\setminus C(S,t-1), \succ'_t) = C(S,t)\setminus C(S,t-1)$. It follows that $C$ is lexicographic for  $(\succ'_1, \ldots ,\succ'_n)$.   

(\textit{Only if part}) For each $t\in \{1,\ldots ,n\}$, let $\mathcal{A}_t$ stand for the collection of all nonempty subsets of $A_t$ with at least $t$ elements.\@ Then, define the choice function $c_t:\mathcal{A}_t\rightarrow A_t$ such that for each choice set $S\in \mathcal{A}_t$,  $c_t(S)=C(S,t)\setminus C(S,t-1)$. Since $C$ satisfies  \textit{gross substitutes}, $c_t$ also satisfies \textit{gross substitutes}.\@ It follows that there is a unique priority ordering $\succ^{*}_t$ such that $c_t(S)=\max\{S\setminus C(S,t-1), \succ^{*}_t\}$.\@ Therefore, if $C$ is lexicographic for some  $(\succ_1, \ldots ,\succ_n)$, then for each $t\in \{1,\ldots ,n\}$, $\succ_t|_{A_t}\: =\:\succ^{*}_t$.
\end{proof}

\subsection{Lexicographic Choice Under Feasibility Constraints}
\label{feasibility}

In some applications, the choice rule of an institution is subject to a feasibility constraint. For example, a firm may encounter a choice set which includes signing the same worker under different terms, such as different salaries as modeled in \cite{kelso_crawford_1982}, and it may not be possible to choose the same worker under several terms even when there is enough capacity (for instance, it is not possible to choose the same worker under different salaries). The matching with contracts model due to \cite{hm2005} introduced a general framework that incorporates such feasibility constraints into the matching problem, which led to several new applications of matching theory such as cadet-branch matching by \cite{sonmez13} and \cite{sonmez13b}, and matching with slot-specific priorities by \cite{ks2012}. In this section, we will show that our baseline model and our baseline properties can be extended to a setup with feasibility constraints, highlighting the distinguishing properties of lexicographic choice rules in a more general setup.
As in the baseline model, let $A$ be a nonempty finite set of $n$ alternatives and let $\mathcal{A}$ denote the set of all \textit{nonempty} subsets of $A$. In addition, let $\mathcal{F}\subseteq \mathcal{A}$ be a nonempty set of \textit{feasible} sets. We assume that $\mathcal{F}$  is \textit{downward closed} in the sense that for each $S\in \mathcal{F}$ and each $S'\subseteq S$, $S'\in \mathcal{F}$.\footnote{In a matching with contracts model with distributional constraints, \cite{goto17} introduce the concept of a ``hereditary" distributional constraint, which implies that $\mathcal{F}$ is downward closed.} We also assume that each singleton is feasible, i.e.\ for each $a\in A$, $\{a\}\in \mathcal{F}$.\footnote{Note that, given downward closedness, this is equivalent to requiring that each alternative belongs to at least one feasible set.}

A (feasibility-constrained) \textbf{choice rule} $C: \mathcal{A} \times \{1,\ldots ,n\} \rightarrow \mathcal{F}$ 
associates with each problem $(S,q)\in \mathcal{A} \times \{1,\ldots ,n\}$, a \textit{nonempty} set of choices $C(S,q)\subseteq S$ which is feasible, i.e.\ $C(S,q)\in \mathcal{F}$, and respects the capacity constraint, i.e.\ $|C(S,q)|\leq q$. Given a choice rule $C$, we denote the set of rejected alternatives at a problem $(S,q)$ by $R(S,q)=S\setminus C(S,q)$.\medskip

Our new framework encompasses the matching with contracts framework in the following way. Suppose that each alternative is a contract consisting of a pair: an agent and a contractual term. Suppose that a choice set is feasible if it includes, for each agent, at most one contract including that agent. It is easy to see that $\mathcal{F}$  is downward closed and it includes the singletons. 

A feasibility-constrained choice rule $C$ is \textbf{(capacity-constrained) lexicographic} if there exists a priority profile $(\succ_1, \ldots ,\succ_n)\in \Pi$ such that for each $(S,q)\in \mathcal{A} \times \{1,\ldots ,n\}$, $C(S,q)$ is obtained by choosing the highest $\succ_1$-priority alternative in $S$, then choosing the highest $\succ_2$-priority alternative among the remaining alternatives that induces a feasible set together with the previously chosen alternative,\footnote{Formally, let $a$ be the highest $\succ_1$-priority alternative in $S$. Let $S'=\{b\in S\setminus \{a\}: \{a,b\}\in \mathcal{F}\}$. Then, the highest $\succ_2$-priority alternative among the remaining alternatives that induces a feasible set together with the previously chosen alternative is $max\{S',\succ_2\}$.} and so on as long as there is a remaining alternative until finally choosing the highest $\succ_q$-priority alternative among the remaining alternatives that induces a feasible set together with the previously chosen alternatives. \medskip

\noindent \textbf{$\mathcal{F}$-capacity-filling:} An alternative is rejected from a choice set at a capacity only if the capacity is full or it is infeasible to choose the alternative. Formally, for each $(S,q)\in \mathcal{A} \times \{1,\ldots ,n\}$ and $a\in S$, if $a\notin C(S,q)$, then either $|C(S,q)|=q$ or $C(S,q)\cup \{a\}\notin \mathcal{F}$.

Let us adopt the convention that for each $S\in \mathcal{A}$, $C(S,0)=\emptyset$. Now, for each capacity $q\in \{1,\ldots, n\}$, $a$ is \textbf{revealed to be $\mathcal{F}$-preferred} to $b$ at $q$, denoted by $a\mathrel{R_q^\mathcal{F}}b$, if there exists $S\in \mathcal{A}$ such that $a,b\notin C(S,q-1)$, and $a\in C(S,q)$ but $b\notin C(S,q)$, although $C(S,q-1)\cup \{b\}\in \mathcal{F}$. We introduce the following property which requires, for each capacity, the revealed preference relation be acyclic.

\noindent \textbf{Capacity-wise strong axiom of revealed preference (CSARP):} For each capacity $q\in \{1,\ldots, n\}$, $\mathrel{R_q^\mathcal{F}}$ is acyclic.

\begin{proposition}
        A feasibility-constrained choice rule is (capacity-constrained) lexicographic if and only if it satisfies \textit{$\mathcal{F}$-capacity-filling},  \textit{monotonicity}, and the \textit{capacity-wise strong axiom of revealed preference}.
\end{proposition} 

\begin{proof}
        
(\textit{Only if part:}) Let $C$ be a feasibility-constrained choice rule that is lexicographic for $(\succ_1, \ldots, \succ_n)$.
Using similar arguments as in the proof of Theorem \ref{thm1}, one can easily verify that $C$ satisfies $\mathcal{F}$-capacity-filling and monotonicity. 

To see that $C$ satisfies CSARP, note  that for each capacity $q\in \{1, \ldots, n\}$ and $a,b\in A$, if $a \mathrel{R_q^\mathcal{F}}b$, then we must have $a \succ_q b$. 
Since   $ \succ_q $ is transitive, $ R_q^\mathcal{F} $ is acyclic.
        
(\textit{If part:}) Let $C$ be a feasibility-constrained choice rule that satisfies $\mathcal{F}$-capacity-filling,  monotonicity, and CSARP. 
It follows from CSARP that for each $q\in \{1, \ldots, n\}$, 
$R_q^\mathcal{F}$ is acyclic. Now, for each capacity $q\in \{1, \ldots, n\}$, let $\succ_q$ be any completion of the transitive closure of $R_q^\mathcal{F}$.
Next, we show that $C$ is lexicographic for $(\succ_1, \ldots, \succ_n)$. 
To see this, we apply induction on capacity $q$. 
Before proceeding, let us introduce some notation. 
For each $S,T\in  \mathcal{A}$ such that $T\subset S$, let $\mathcal{F}(S|_{T})$ be the set of alternatives in $S\setminus T$ that induce a feasible set together with the alternatives in $T$, i.e. $\mathcal{F}(S|_{T})=\{a\in S\setminus T : T\cup \{a \}\in  \mathcal{F}  \}$.

First, we show that for each $S\in \mathcal{A}$, $C(S,1)=max(S, \succ_1)$. By contradiction suppose that although $a=max(S, \succ_1)$, we have $C(S,1)=b$, where $a\not = b$. 
Since $C(S,1)=b$ and $a\in S$, it follows that $b\mathrel{R_1^\mathcal{F}}a $, which contradicts that $a=max(S, \succ_1)$.
Next, assume that for some $q\in \{2, \ldots, n\}$, we have 
for each $S\in \mathcal{A}$ and $q'<q$, $C(S, q')$ coincides with the lexicographic choice for $(\succ_1, \ldots, \succ_{q-1})$. 
Now, we show that for each $S\in \mathcal{A}$, 
$C(S,q)\setminus C(S, q-1)=max(\mathcal{F}(S|_{C(S, q-1)}), \succ_q)$.
First, let $a=max(\mathcal{F}(S|_{C(S, q-1)}), \succ_q)$. 
It follows that  $a\notin C(S, q-1)$  and $C(S, q-1)\cup \{a\}\in \mathcal{F}$. 
By contradiction, suppose that $a\notin C(S, q)$. Since  $C(S, q-1)\cup \{a\}\in \mathcal{F}$, it follows from $\mathcal{F}$-capacity-filling that there exists $x\in C(S,q)\setminus C(S, q-1)$ such that $x\not = a$. 
Now, since $C$ satisfies monotonicity, $x\notin C(S, q-1)$, and  
since $x\in C(S,q)$, $C(S, q-1)\cup \{x\}\in \mathcal{F}$.
Therefore, we have $x\mathrel{R_q^\mathcal{F}}a $, but this contradicts that $a=max(\mathcal{F}(S|_{C(S, q-1)}), \succ_q)$. 
Thus, we conclude that $C$ is lexicographic for $(\succ_1, \ldots, \succ_n)$.
\end{proof}

\section{Lexicographic Deferred Acceptance Mechanisms}
\label{allocation_model}

Let $N$ denote a finite set of agents, $\left| N\right|=n\geq 2$. Let $\mathcal{A}$ be the collection of all nonempty subsets of $N$. Let $O$ denote a finite set of objects. Each agent $i\in N$ has a complete, transitive, and anti-symmetric preference relation $R_{i}$ over $O\cup \{\emptyset\}$, where $\emptyset$ is the null object representing the option of receiving no object (or receiving an outside option). Given $x,y\in O\cup \{\emptyset\}$, $x\mathbin{R_{i}}y$ means that either $x=y$ or $x\neq y$ and agent $i$ prefers $x$ to $y$. If agent $i$ prefers $x$ to $y$, we write $x\mathbin{P_{i}}y$. 
Let $\mathcal{R}$ denote the set of all preference relations over $O\cup \{\emptyset\}$, and $\mathcal{R}^{N}$ the set of all preference profiles $R=(R_{i})_{i\in N}$ such that for all $i\in N$, $R_{i}\in \mathcal{R}$.

An allocation problem with capacity constraints, or simply a \textbf{problem}, consists of a preference profile $R\in \mathcal{R}^{N}$ and a capacity profile $q=(q_x)_{x\in O\cup \{\emptyset\}}$ such that for each object $x\in O$, $q_x\in \{0,1,\ldots,n\}$ and $q_{\emptyset }=n$ so that the null object has enough capacity to accommodate all agents. Let $\mathcal{P}$ denote the set of all problems. Given a problem $(R,q)\in \mathcal{P}$, an object $x$ is \textbf{available} at the problem if $q_x>0$. 

Given a capacity profile $q=(q_x)_{x\in O\cup \{\emptyset\}}$, an allocation assigns to each agent exactly one object in $O\cup \{\emptyset\}$ taking capacity constraints into account. Formally, an \textbf{allocation} at $q$ is a list $a=(a_{i})_{i\in N}$ such that for each $i\in N$, $a_{i}\in O\cup \{\emptyset\}$ and no object $x\in O\cup \{\emptyset\}$ is assigned to more than $q_x$ agents. Let $M(q)$ denote the set of all allocations at $q$. 

Given an allocation $a=(a_{i})_{i\in N}$, a preference profile $R$, and an object $x\in O\cup \{\emptyset\}$, let $D_x(a,R)=\{i\in N: x\mathrel{P_i}a_i\}$ denote the \textbf{demand} for $x$ at $(a,R)$, which is the set of agents who prefer $x$ to their assigned object. 

A \textbf{mechanism} is a function $\varphi:\mathcal{P}\rightarrow \bigcup_{q} M(q)$ such that for each allocation problem $(R,q)\in \mathcal{P}$, $\varphi(R,q)\in M(q)$. 

For mechanisms, we introduce the following property, which we call \textit{the irrelevance of satisfied demand}. Consider an arbitrary problem and the allocation chosen by the mechanism at that problem. Suppose that the capacity of an object is increased. Now, some of the agents who prefer that object to their assignments at the initial allocation may receive the object due to the capacity increase.  \textit{The irrelevance of satisfied demand} requires that the set of agents who receive the object due to the capacity increase does not depend on the set of agents who initially receive the object.  
In other words, for two problems with the same capacity, if the demands for an object are the same (note that the set of agents who receive the object at those problems may be different), then whenever the capacity of the object increases, the sets of agents who receive the object due to the capacity increase should be the same for the two problems.

Formally, for each $x\in O$, let $1_x$ be the capacity profile which has $1$ unit of $x$ and nothing else. A mechanism $\varphi$ satisfies \textbf{the irrelevance of satisfied demand} if for each pair of problems $(R,q)$ and $(R',q)$ and each object $x\in O$, if $D_x(\varphi(R,q),R)=D_x(\varphi(R',q),R')$, then $D_x(\varphi(R,q+1_x),R)=D_x(\varphi(R',q+1_x),R')$.

A \textbf{(capacity-constrained) lexicographic choice structure} $\mathcal{C}=(C_x)_{x\in O}$ associates each object $x\in O$ with a lexicographic choice rule $C_x: \mathcal{A} \times \{1,\ldots ,n\} \rightarrow \mathcal{A}$.
Next, we present \textbf{the (capacity-constrained) lexicographic deferred acceptance algorithm based on $\mathcal{C}$}. For each problem $(R,q)\in \mathcal{P}$, the algorithm runs as follows:\\
\textit{Step 1:} Each agent applies to his favorite object in $O$. Each object $x\in O$ such that $q_x> 0$ temporarily accepts the applicants in $\mathcal{C}_x(S_x,q_x)$ where $S_x$ is the set of agents who applied to $x$, and rejects all the other applicants. Each object $x\in C$ such that $q_x= 0$ rejects all applicants.\\
\textit{Step $r\geq 2$:} Each applicant who was rejected at step $r-1$ applies to his next favorite object in $O$. For each object $x\in O$, let $S_{x,r}$ be the set consisting of the agents who applied to $x$ at step $r$ and the agents who were temporarily accepted by $x$ at Step $r-1$. Each object $x\in O$ such that $q_x> 0$ accepts the applicants in $C_x(S_{x,r},q_x)$ and rejects all the other applicants. Each object $x\in O$ such that $q_x= 0$ rejects all applicants.

\noindent The algorithm terminates when each agent is accepted by an object. The allocation where each agent is assigned the object that he was accepted by at the end of the algorithm is called the $C$-lexicographic Deferred Acceptance allocation at $(R,q)$, denoted by $DA^{\mathcal{C}}(R,q)$.

\noindent\textbf{Lexicographic deferred acceptance mechanisms:} A mechanism $\varphi$ is a \textit{lexicographic deferred acceptance mechanism} if there exists a lexicographic choice structure $\mathcal{C}$ such that for each $(R,q)\in \mathcal{P}$, $\varphi(R,q)=DA^{\mathcal{C}}(R,q)$.

\cite{ehlers_klaus_2014b}, in their Theorem $3$, characterize deferred acceptance mechanisms based on a choice structure satisfying \textit{capacity-filling}, \textit{gross substitutes}, and \textit{monotonicity}, with the following properties of mechanisms: unavailable-type-invariance (if the positions of the unavailable types are shuffled at a profile, then the allocation should not change); weak non-wastefulness (no agent receives the null object while he prefers an object that is not exhausted to the null object),\footnote{The stronger version of this property, namely \textit{non-wastefulness}, requires that no agent prefers an object that is not exhausted to his assigned object. Note that \textit{capacity-filling} and \textit{non-wastefulness} are similar in spirit, yet, \textit{capacity-filling} is a property of a choice rule while \textit{non-wastefulness} is a property of a mechanism.} resource-monotonicity (increasing the capacities of some objects does not hurt any agent), truncation-invariance (if an agent truncates his preference relation in such a way that his allotment remains acceptable under the truncated preference relation, then the allocation should not change), and strategy-proofness (no agent can benefit by misreporting his preferences). Next, we formally introduce these properties and state Theorem $3$ of \cite{ehlers_klaus_2014b}.

\noindent \textbf{Unavailable-Type-Invariance:} Let $(R,q)\in \mathcal{P}$ and $R'\in\mathcal{R}^N$. If for each $i\in N$ and each pair of available objects $x,y\in O$ ($q_x>0$, $q_y>0$) we have [$x\mathrel{R_i}y$ if and only if $x\mathrel{R_i'}y$], then $\varphi(R,q)=\varphi(R',q)$.

\noindent \textbf{Weak Non-Wastefulness:} For each $(R,q)\in\mathcal{P}$, each $x\in O$ such that $q_x>0$, and each $i\in N$, if $x\mathbin{P_i}\varphi_i(R,q)$ and $\varphi_i(R,q)=\emptyset$, then $|\{j\in N:\varphi_j(R,q)=x\}|=q_x$.

\noindent \textbf{Resource-Monotonicity:} For each $R\in\mathcal{R}^{N}$, and each pair of capacity profiles $(q,q')$, if for each $x\in O$, $q_x\leq q_x'$, then for each $i\in N$, $\varphi_{i}(R,q')\mathbin{R_i}\varphi_{i}(R,q)$.

\noindent \textbf{Truncation-Invariance:} Let $(R,q)\in \mathcal{P}$ and $R'\in\mathcal{R}^N$. If for each $i\in N$ and each pair of objects $x,y\in O$ we have [$x\mathrel{R_i}y$ if and only if $x\mathrel{R_i'}y$] and $\varphi_i(R,q)\mathrel{R'_i}\emptyset$, then $\varphi(R,q)=\varphi(R',q)$.

\noindent \textbf{Strategy-proofness:} For each $(R,q)\in\mathcal{P}$, each $i\in N$, and each $R_i'\in \mathcal{R}$, $\varphi_{i}(R,q)\mathbin{R_i}\varphi_{i}((R_i',R_{-i}),q)$.

\noindent \textbf{Theorem 3 of Ehlers and Klaus(2016):} \textit{A mechanism is a deferred acceptance mechanism based on a choice structure satisfying capacity-filling, gross substitutes, and monotonicity if and only if it satisfies unavailable-type-invariance, weak non-wastefulness, resource-monotonicity, truncation-invariance, and strategy-proofness.}

The following impossibility result shows that \textit{the irrelevance of satisfied demand} is too strong: there is no mechanism which satisfies it together with the above desirable properties.

\begin{proposition}
\label{impossibility}
Suppose that there are at least three objects, $|O|\geq 3$. There is no mechanism which satisfies unavailable-type-invariance, weak non-wastefulness, resource-monotonicity, truncation-invariance, strategy-proofness, and the irrelevance of satisfied demand.
\end{proposition}

\begin{proof}
Suppose that there exists such a mechanism, say $\varphi$, which satisfies all the properties in the statement except for \textit{the irrelevance of satisfied demand}. We will show that it must violate \textit{ the irrelevance of satisfied demand}. By Theorem $3$ of \cite{ehlers_klaus_2014b}, $\varphi$ is a deferred acceptance mechanism based on a choice structure $\mathcal{C}=(C_x)_{x\in O}$ which satisfies \textit{capacity-filling}, \textit{gross substitutes}, and \textit{monotonicity}. 

Let $i,j\in N$ be two distinct agents. We first claim that there exist two distinct objects $a,b\in O$ such that $i\in C_a(\{i,j\},1)\cap C_b(\{i,j\},1)$ and $j\notin C_a(\{i,j\},1)\cup C_b(\{i,j\},1)$. That is, when there is only one unit of $a$ or $b$, $i$ is chosen but $j$ is not from $\{i,j\}$. To see this, let $x,y,z\in O$ be three distinct objects. By \textit{capacity-filling}, either $\{i\}= C_x(\{i,j\},1)$ or $\{j\}= C_x(\{i,j\},1)$. Without loss of generality, suppose that $\{i\}= C_x(\{i,j\},1)$. Again by \textit{capacity-filling}, either $\{i\}= C_y(\{i,j\},1)$ or $\{j\}= C_y(\{i,j\},1)$. If $\{i\}= C_y(\{i,j\},1)$, then we are done. Otherwise, by \textit{capacity-filling}, either $\{i\}= C_z(\{i,j\},1)$ or $\{j\}= C_z(\{i,j\},1)$, and in either case, we are done.

So, suppose that there exist two distinct objects $a,b\in O$ such that $i\in C_a(\{i,j\},1)\cap C_b(\{i,j\},1)$ and $j\notin C_a(\{i,j\},1)\cup C_b(\{i,j\},1)$. Let the preference profiles $R$ and $R'$ be such that every agent other than $i$ and $j$ find any object unacceptable and $R_i, R_j, R'_i$ and $R'_j$ are as depicted below.  
\begin{figure}[H]
        
        \begin{center}
        \begin{tabular}{cccccc}
        $R_{i}$& $R_{j}$& & &$R'_{1}$& $R'_{2}$ \\
\cline{1-2} \cline{5-6}
        $a$&$b$& & &$a$&$a$\\
        $b$&$a$& & &$b$&$b$\\
        $\emptyset$&$\emptyset$&  & &$\emptyset$&$\emptyset$ 
\end{tabular}
\end{center}

\end{figure}
Let $q$ be such that $q_b=1$ and $q_x=0$ for any $x\in O\setminus \{b\}$. Let $q'$ be such that $q'_a=q'_b=1$ and $q'_x=0$ for any $x\in O\setminus \{a,b\}$. Since $\varphi$ is a deferred acceptance mechanism based on $\mathcal{C}=(C_x)_{x\in O}$, $D_a(\varphi(R,q),R)=D_a(\varphi(R',q),R')=\{i,j\}$. However, $D_a(\varphi(R,q'),R)=\emptyset$ and $D_a(\varphi(R',q'),R')=\{j\}$, implying that $\varphi$ violates \textit{the irrelevance of satisfied demand}.

\end{proof}

Since \textit{the irrelevance of satisfied demand} is too strong, we consider the following weakening of it which requires that \textit{at any problem where there is only one available object}, the set of agents who receive the object due to a capacity increase does not depend on the set of agents who initially receive the object. 

Formally, a mechanism $\varphi$ satisfies \textbf{the weak irrelevance of satisfied demand} if for any pair of problems $(R,q)$ and $(R',q)$ and each object $x\in O$ such that for each $y\in O\setminus \{x\}$, $q_y=0$, $D_x(\varphi(R,q),R)=D_x(\varphi(R',q),R')$ implies $D_x(\varphi(R,q+1_x),R)=D_x(\varphi(R',q+1_x),R')$.

Our next result shows that \textit{the weak irrelevance of satisfied demand} together with the above properties characterize lexicographic deferred acceptance mechanisms.  

\begin{proposition}
\label{ek_corollary}
A mechanism is a lexicographic deferred acceptance mechanism if and only if it satisfies unavailable-type-invariance, weak non-wastefulness, resource-monotonicity, truncation-invariance, strategy-proofness, and the weak irrelevance of satisfied demand.
\end{proposition}

\begin{proof}
The following notation will be helpful. For each $x\in O$, let $R^x$ be a preference relation such that $x$ is top-ranked and $\emptyset$ is second-ranked. For each $S\in \mathcal{A}$ that is nonempty, let $R_S^x$ be a preference profile such that for each $i\in S$, $(R_S^x)_i=R^x$, and for each $j\notin S$, $(R_S^x)_j$ top-ranks $\emptyset$. For each $x\in O$ and $l\in \{0, \ldots ,n\}$, let $l_x$ denote the capacity profile where $x$ has capacity $l$ and every other object has capacity zero.

Let $\varphi$ be a mechanism satisfying the properties in the statement of the theorem. Let $\mathcal{C}=(C_x)_{x\in O}$ be defined as follows. For each $x\in O$, $S\in \mathcal{A}$, and  $l\in \{0, \ldots ,n\}$, $C_x(S,l)=\{i\in S:\varphi_i(R_S^x,l_x)=x\}$. This choice structure is the same as the one constructed in the proof of Theorem $3$ of \cite{ehlers_klaus_2014b}.   

By \textit{weak non-wastefulness}, $C_x$ satisfies \textit{capacity-filling}. By \textit{resource-monotonicity}, $C_x$ satisfies \textit{monotonicity}. By Lemma 2 of \cite{ehlers_klaus_2014b}, $C_x$ satisfies \textit{gross substitutes}. By Theorem $3$ of \cite{ehlers_klaus_2014b}, $\varphi$ is a deferred acceptance mechanism based on $\mathcal{C}$. It is easy to see that, since $\varphi$ satisfies \textit{the irrelevance of satisfied demand}, for each $x\in O$, $C_x$ satisfies \textit{the irrelevance of accepted alternatives}. Thus, $\mathcal{C}$ is a lexicographic choice structure and $\varphi$ is a lexicographic deferred acceptance mechanism.

Let $\varphi$ be a lexicographic deferred acceptance mechanism. We will show that it satisfies \textit{irrelevance of satisfied demand}. The other properties follow from Theorem $3$ of \cite{ehlers_klaus_2014b}. Let $\mathcal{C}=(C_x)_{x\in O}$ be a lexicographic choice structure such that $\varphi=DA^\mathcal{C}$. Let $(R,q),(R',q) \in  \mathcal{P}$ and $x\in O$ be such that for each $y\in O\setminus \{x\}$, $q_y=0$ and let $T\equiv D_x(DA^\mathcal{C}(R,q),R)=D_x(DA^\mathcal{C}(R',q),R')$. Let $C_x$ be lexicographic for the priority profile $(\succ_1, \ldots ,\succ_n)\in \Pi$. Let $S(R)$ and $S(R')$ be the sets of agents who prefer $x$ to $\emptyset$ at $R$ and at $R'$, respectively. It is easy to see that $DA^\mathcal{C}(R,q)=C_x(S(R),q)$, $DA^\mathcal{C}(R',q)=C_x(S(R'),q)$, and $T=S(R)\setminus C_x(S(R),q)=S(R')\setminus C_x(S(R'),q)$. Let $i\in T$ be the agent who is highest ranked according to $\succ_{q_x+1}$ in $T$. Clearly, $DA^\mathcal{C}(R,q+1_x)=DA^\mathcal{C}(R,q)\cup \{i\}$ and $DA^\mathcal{C}(R',q+1_x)=DA^\mathcal{C}(R',q)\cup \{i\}$. Hence, $D_x(DA^\mathcal{C}(R,q+1_x),R)=D_x(DA^\mathcal{C}(R',q+1_x),R')=T\setminus \{i\}$.
\end{proof}

\begin{remark}
In Appendix \ref{appendix:example}, we provide an example of a mechanism which satisfies all the properties in the statement of Proposition~\ref{ek_corollary} except for \textit{the irrelevance of satisfied demand}, and therefore which is not a lexicographic deferred acceptance mechanism.
\end{remark}

\section{Implications for School Choice in Boston}
\label{boston}

In the Boston school choice system, there are two different priority orderings at each school: a \textit{walk-zone priority ordering}, which gives priority to the school's neighborhood students over all the other students, and an \textit{open priority ordering} which does not give priority to any student for being a neighborhood student. The Boston school district aims to assign half of the seats of each school based on the walk-zone priority ordering and the other half based on the open priority ordering. To achieve this aim, given the capacity, each school chooses students in a lexicographic way according to a priority profile where half of the priority orderings is the \textit{walk-zone priority ordering} and the other half is the \textit{open priority ordering}.

In a recent study, \cite{dur2016} note that two priority profiles with the same numbers of walk-zone and open priority orderings, but with different precedence orders of the priority orderings, may result in different choices under a lexicographic choice procedure. Starting with this observation, \cite{dur2016} compare four different choice rules, one of which is the one used in the Boston school district,  in terms of how much they are biased for or against the neighbourhood students.\@ In this section, we will consider these alternative choice rules from a different perspective.\@ We will show that, although these choice rules are all based on a ``lexicographic procedure" at each capacity, only one of them satisfies all the characterizing properties in Theorem \ref{thm1}, and therefore only one of them is actually a \textit{(capacity-constrained) lexicographic choice rule}. 

In order to put the four choice rules in a formal context, let us consider the following class of choice rules which is larger than the class of lexicographic choice rules. We say that a choice rule is \textbf{capacity-wise lexicographic} if there exists a list of priority orderings for each capacity level (the number of priority orderings is the same as the capacity), and at each capacity, the rule operates based on the associated list of priority orderings in a lexicographic way. For a capacity-wise lexicographic choice rule, unlike a lexicographic choice rule, the lists for different capacity levels are not necessarily related.

The capacity-wise lexicographic choice rules that can serve the Boston school district's purpose are the choice rules for which, at each capacity, the associated list consists of only the walk-zone priority ordering and the open priority ordering, and the absolute difference between the numbers of walk-zone and open priority orderings in the list is at most one. We formalize this property as follows. 

Let $\succ^w$ and $\succ^o$ be walk-zone and open priority orderings. We say that a capacity-wise lexicographic choice rule satisfies the \textbf{Boston requirement for} ($\succ^w$,$\succ^o$) if for each capacity $q$, the associated list of priority orderings $(\succ_1,\ldots ,\succ_q)$ is such that

\begin{enumerate}
        \item[i.] for each $l\in \{1,\ldots ,q\}$, $\succ_l\in \{\succ^w,\succ^o\}$,
        \item[ii.] 
        difference between the number of $\succ^w$-priorities and $\succ^o$-priorities is at most one, i.e.
        $\Big|\sum_{i=1}^q{1_{\succ^w}(\succ_i)}-\sum_{i=1}^q{1_{\succ^o}(\succ_i)}\Big|\leq 1.\footnote{$1_x(y)$ is the indicator function which has the value $1$ if $x=y$ and $0$ otherwise. }$
\end{enumerate}

Now, it turns out that the following class of capacity-wise lexicographic choice rules are the only rules satisfying our set of properties together with the Boston requirement for ($\succ^w$,$\succ^o$).

\begin{proposition}
        \label{prop_boston}
        A capacity-wise lexicographic choice rule satisfies \textit{capacity-filling}, \textit{gross substitutes}, \textit{monotonicity}, the \textit{capacity-wise weak axiom of revealed preference}, and the \textit{Boston requirement for} ($\succ^w$,$\succ^o$) if and only if it is (capacity-constrained) lexicographic for a priority profile $(\succ_1, \ldots ,\succ_n)$ such that
        
        \begin{enumerate}
                \item[i.] for each $l\in \{1,\ldots ,n\}$, $\succ_l\in \{\succ^w,\succ^o\}$,
                \item[ii.] for each $l$ that is odd, $\succ_l=\succ^w$ if and only if $\succ_{l+1}=\succ^o$.
        \end{enumerate}
\end{proposition} 

\begin{proof}
By Theorem~\ref{thm1}, a choice rule satisfying the properties must be lexicographic. The rest is straightforward.
\end{proof}

\noindent Some examples of priority profiles satisfying (i) and (ii) in the statement of Proposition \ref{prop_boston} are $(\succ^w, \succ^o,\succ^w,\succ^o, \ldots)$,
$(\succ^o, \succ^w,\succ^o,\succ^w, \ldots)$, and $(\succ^w, \succ^o,\succ^o,\succ^w, \succ^w,\succ^o, \succ^w,\succ^o,\ldots)$. Some examples that violate (ii) are $(\succ^w, \succ^o,\succ^o,\succ^w, \succ^o,\succ^o, \ldots)$ and 
$(\succ^o, \succ^w,\succ^w,\succ^o, \succ^w,\succ^w, \ldots)$.

Four plausible choice rules stand out from the analysis of \cite{dur2016}, one of which is currently in use in Boston (Open-Walk choice rule). \cite{dur2016} compare the below four choice rules in terms of how much they are biased for or against the neighbourhood students. We will compare the four choice rules with respect to our set of choice rule properties.

\begin{enumerate}
        \item \textit{Walk-Open Choice Rule:} At each capacity, the first half of the priority orderings in the list are the walk-zone priority ordering and the last half are the open priority ordering.
        \item \textit{Open-Walk Choice Rule:} At each capacity, the first half of the priority orderings in the list are the open priority ordering and the last half are the walk-zone priority ordering.
        \item \textit{Rotating Choice Rule:} At each capacity, the first priority ordering in the list is the walk-zone priority ordering, the second is the open priority ordering, the third is the walk-zone priority ordering, and so on.
        \item \textit{Compromise Choice Rule:} At each capacity, the first quarter of the priority orderings in the list are the walk-zone priority ordering, the following half of the priority orderings in the list are the open priority ordering, and the last quarter are again the walk-zone priority ordering.
\end{enumerate}

To be precise, let us introduce the following procedures to accommodate the cases where the capacity is not divisible by two or four.

\begin{itemize}
        \item \textit{Walk-Open Choice Rule:} If the capacity $q$ is an odd number, the first $\frac{q+1}{2}$ are the walk-zone priority ordering.
        \item \textit{Open-Walk Choice Rule:} If the capacity $q$ is an odd number, the first $\frac{q+1}{2}$ are the open priority ordering.
        \item \textit{Compromise Choice Rule:} If the capacity $q$ is not divisible by four, let $q=q'+k$ for some $q'$ that is divisible by $4$ and some $k\in \{1,2,3\}$. If $k=1$, let the first $\frac{q'}{4}+1$ orderings be the walk-zone priority ordering, the following $\frac{q'}{2}$ orderings be the open priority ordering, and the last $\frac{q'}{4}$ orderings be the walk-zone priority ordering. If $k=2$, let the first $\frac{q'}{4}+1$ orderings be the walk-zone priority ordering, the following $\frac{q'}{2}+1$ orderings be the open priority ordering, and the last $\frac{q'}{4}$ orderings be the walk-zone priority ordering. If $k=3$, let the first $\frac{q'}{4}+1$ orderings be the walk-zone priority ordering, the following $\frac{q'}{2}+1$ orderings be the open priority ordering, and the last $\frac{q'}{4}+1$ orderings be the walk-zone priority ordering.
\end{itemize}

Note that all of the above rules satisfy the Boston requirement for ($\succ^w$,$\succ^o$). Since all of the rules are capacity-wise lexicographic, they satisfy \textit{capacity-filling} and \textit{gross substitutes}. 
It follows from the only if part of Proposition~\ref{prop_boston} that, among these four choice rules, only the Rotating Choice Rule satisfies \textit{capacity-filling}, \textit{gross substitutes}, \textit{monotonicity}, and \textit{the irrelevance of accepted alternatives}. 
However, it is not clear if the other three rules are not lexicographic under variable capacity constraints because they fail to satisfy \textit{monotonicity}, \textit{the irrelevance of accepted alternatives} or both.
Next, we show that the other three rules satisfy \textit{monotonicity}, but they violate \textit{the irrelevance of accepted alternatives}.
To show  that these rules satisfy \textit{monotonicity},  we first
provide an auxiliary condition that is easier to verify and sufficient for \textit{monotonicity}. Next we introduce this condition and prove that it is sufficient for \textit{monotonicity}.

Let $\pi=(\succ_1, \ldots ,\succ_{q})$ and $\pi'=(\succ'_1, \ldots ,\succ'_{q+1})$ be priority lists of size $q$ and $q+1$, respectively. We say that $\pi'$ is \textbf{obtained by insertion} from $\pi$ if there exists $k\in \{1,\ldots ,q+1\}$ such that  $\succ'_l=\succ_l$ for each $l<k$, and  $\succ'_l=\succ_{l-1}$ for each $l >k$. Note that when $\pi'$ is obtained by insertion from $\pi$, a new priority ordering is inserted into the list of priority orderings in $\pi$, by keeping relative order of the other priority orderings in the list the same. It is possible that the new ordering is inserted in the very beginning or in the very end of the list. 

\begin{lemma}
\label{insertion}
Let $C$ be a capacity-wise lexicographic choice rule. The choice rule $C$ is \textit{monotonic} if for each $q\in \{2,\ldots, n\}$, the priority list for $q$ is obtained by insertion from the priority list for $q-1$.  
\end{lemma}

\begin{proof}
Let $(S,q)\in \mathcal{A} \times \{1,\ldots ,n-1\}$. Let $\pi=(\succ_1, \ldots ,\succ_{q})$ be the list for capacity $q$. Let $a\in C(S,q)$. Suppose that, in the lexicographic choice procedure, $a$ is chosen at the $t$'th step, i.e.\ $a$ is chosen based on $\succ_t$. 

Let $\pi'=(\succ'_1, \ldots ,\succ'_{q+1})$ be the list for capacity $q+1$ that is obtained by an insertion from $\pi$. Let $k\in \{1,\ldots ,q+1\}$ be such that  $\succ'_l=\succ_l$ for each $l<k$, and  $\succ'_l=\succ_{l-1}$ for each $l >k$.

Now, consider the problem $(S,q+1)$. If $t<k$, clearly $a$ is still chosen at the $t$'th step of the lexicographic choice procedure and thus $a\in C(S,q+1)$. Suppose that $t\geq k$. The rest of the proof is by induction. First, suppose that $t=k$. Note that at Step $k$ of the choice procedure for the problem $(S,q+1)$, the choice is made based on the inserted priority ordering and at Step $k+1$, the choice is made based on $\succ_t$. Then, $a$ is either chosen at Step $k$, or at Step $k+1$, the set of remaining alternatives is a subset of the set of remaining alternatives at Step $t$ of the choice procedure for $(S,q)$ where $a$ is chosen, in which case $a$ is still chosen. Thus, $a\in C(S,q+1)$.

Now, suppose that $t>k$ and each alternative that is chosen at a step $t'<t$ of the choice procedure at $(S,q)$ is also chosen at $(S,q+1)$. Then, $a$ is either chosen before step $t+1$ of the choice procedure for $(S,q+1)$, or at Step $t+1$, the set of remaining alternatives is a subset of the set of remaining alternatives at Step $t$ of the choice procedure for $(S,q)$ where $a$ is chosen, in which case $a$ is still chosen. Thus, $a\in C(S,q+1)$.   
\end{proof}

\begin{proposition}
\label{propp}
All of the four rules satisfy \textit{capacity-filling}, \textit{gross substitutes}, and \textit{monotonicity}, but only the rotating choice rule satisfies \textit{the irrelevance of accepted alternatives} and only the rotating choice rule is (capacity-constrained) lexicographic.
\end{proposition}

\begin{proof}
Each rule is capacity-wise lexicographic (lexicographic for a given capacity) and therefore satisfies \textit{capacity-filling} and \textit{gross substitutes}. 
Moreover, it is easy to see that each of the four choice rules satisfies the insertion property, so \textit{monotonicity} follows from Lemma~\ref{insertion}.

As for \textit{the irrelevance of accepted alternatives}, first consider $(\succ_1, \ldots ,\succ_n)\in \Pi$ such that the first priority ordering in the list is $\succ^w$, the second is $\succ^o$, the third is $\succ^w$, and so on. The rotating choice rule is clearly lexicographic for $(\succ_1, \ldots ,\succ_n)$. Moreover, by Theorem~\ref{thm1}, it satisfies \textit{the irrelevance of accepted alternatives}. We will show that each of the other three choice rules violates \textit{the irrelevance of accepted alternatives}. 

\noindent \textit{Walk-Open Choice Rule:} Let $A=\{a,b,c,d,e\}$. Let $\succ^w$ be defined as $a\succ^w b \succ^w c \succ^w d \succ^w e$ and $\succ^o$ be defined as $e\succ^o b \succ^o d \succ^o c \succ^o a$. Note that $R(\{a,c,d,e\},2)=R(\{a,b,c,d\},2)=\{c,d\}$. However, $R(\{a,c,d,e\},3)=\{d\}$ and $R(\{a,b,c,d\},3)=\{c\}$, and therefore $C$ violates \textit{the irrelevance of accepted alternatives}.    

\noindent \textit{Open-Walk Choice Rule:} Can be shown by interchanging the orderings for $\succ^w$ and $\succ^o$ in the previous example.    

\noindent \textit{Compromise Choice Rule:} Let $A=\{a,b,c,d,x,y\}$. Let $\succ^w$ be defined as $a\succ^w b \succ^w c \succ^w d \succ^w x \succ^w y$ and $\succ_o$ be defined as $b\succ^o c \succ^o y \succ^o x \succ^o d$. Note that $R(\{a,b,c,x,y\},3)=R(\{a,b,d,x,y\},3)=\{x,y\}$. However, $R(\{a,b,c,x,y\},4)=\{y\}$ and $R(\{a,b,d,x,y\},4)=\{x\}$, and therefore $C$ violates \textit{the irrelevance of accepted alternatives}.    
\end{proof}

\begin{remark}
Note that the particular procedures we introduced to accommodate the cases where the capacity is not divisible by two or four are not crucial for the proof of Proposition~\ref{propp}. For the other procedures (for example, for the walk-open choice rule, the extra priority when the capacity is odd can alternatively be set to be the open priority ordering), the examples in the proof can simply be modified to show that \textit{the irrelevance of accepted alternatives} is still violated. 
\end{remark}

It follows from our Proposition \ref{propp} that if \textit{the irrelevance of accepted alternatives} or having a lexicographic representation under variable capacity constraints is deemed desirable, then the rotating choice rule should be selected since it is the only choice rule among the four plausible choice rules that satisfies \textit{the irrelevance of accepted alternatives} together with \textit{capacity-filling}, \textit{gross substitutes}, and \textit{monotonicity}.

Another interpretation of our Proposition \ref{propp} is the following. First of all, note that in \cite{dur2016}, the capacity is fixed and a choice rule is defined given a capacity. On the other hand, the capacity is allowed to vary and a choice rule has to specify which alternatives are chosen from each choice set \textit{at each possible capacity} in our approach, which is the fundamental difference between \cite{dur2016} and our study.\@ The fact that we allow the capacity to vary and we require a choice rule to respond to changes in the capacity, allows us to define desirable properties of choice rules that address how it should respond to changes in the capacity, such as \textit{monotonicity} and \textit{the irrelevance of accepted alternatives}. Proposition \ref{propp} shows that, although each of the four rules in \cite{dur2016} operate based on a lexicographic procedure when we fix the capacity, in a variable capacity framework only one of them satisfies \textit{the irrelevance of accepted alternatives} and therefore only one of them is capacity-constrained lexicographic under variable capacities, i.e., there exists a priority profile, which has as many priority orderings as the maximum possible capacity, such that at each capacity, the rule operates based on a lexicographic procedure \textit{with respect to the same priority profile}. 

\section{Conclusion}
\label{discussion}
Our formulation of a choice rule and the properties that we consider take into account that the capacity may vary. When designing choice rules especially for resource allocation purposes, such as in school choice, a designer may be interested in choice rules that respond to changes in capacity. In that framework, our Theorem~\ref{thm1} shows that \textit{capacity-filling}, \textit{gross substitutes}, \textit{monotonicity}, and \textit{the irrelevance of accepted alternatives} are altogether satisfied only by (capacity-constrained) lexicographic choice rules, which identifies the properties that distinguish lexicographic choice rules from other plausible choice rules. Besides providing an axiomatic foundation for lexicographic choice rules, this finding may be helpful in applications to select among plausible choice rules, as we have illustrated in Section~\ref{boston}, and also to understand characterizing properties of popular resource allocation mechanisms, as we have illustrated in Section~\ref{allocation_model}.

\bibliographystyle{chicago}
\bibliography{boston}

\begin{appendix}
\section*{Appendix}

\section{Necessity of the Properties in Lemma \ref{lemma_rejection_monotonicity}}
\label{necessity_lemma_1}

The following example shows that \textit{capacity-filling} is necessary for the implication, that is, there exists a choice rule which satisfies \textit{monotonicity} and \textit{CWARP}, but not \textit{the irrelevance of accepted alternatives}.

\begin{example}
Let $A=\{a,b,c\}$. Let $\succ$ be defined as $a\succ b \succ c $. Let the choice rule $C$ be defined as follows. For each problem $(S, q)$, if $a\in S$, then $C(S,q)=\{a\}$; if $a\notin S$, then $C$ chooses the highest $\succ$-priority alternatives from $S$ until $q$ alternatives are chosen or no alternative is left. Note that $C$ clearly violates \textit{capacity-filling} and satisfies \textit{monotonicity} and \textit{CWARP}. Now, let $S=\{a,c\}$ and $S'=\{b,c\}$. Note that $R(S,1)=R(S',1)=\{c\}$, $ C(S,2)\cap R(S,1)=\emptyset$, $C(S',2)\cap R(S',1)=\{c\}$, and therefore $C(S,2)\cap R(S,1)\neq C(S',2)\cap R(S',1)$, implying that $C$ violates \textit{the irrelevance of accepted alternatives}. 
\end{example}

The following example shows that \textit{monotonicity} is necessary for the implication, that is, there exists a choice rule which satisfies \textit{capacity-filling} and \textit{CWARP}, but not \textit{the irrelevance of accepted alternatives}.

\begin{example}
Let $A=\{a,b,c,d\}$. Let $\succ$ and $\succ'$ be defined as $a\succ b \succ c \succ d $ and $b\succ' c \succ' d \succ' a$. Let the choice rule $C$ be defined as follows. For each problem $(S, 1)$, $C$ chooses the highest $\succ$-priority alternatives from $S$; and for each problem $(S, q)$ such that $q>1$, $C$ chooses the highest $\succ'$-priority alternatives from $S$ until $q$ alternatives are chosen or no alternative is left. Note that $C$ clearly satisfies \textit{capacity-filling} and \textit{CWARP}. $C$ violates \textit{monotonicity} because, for instance, $a\in C(\{a,b,c\},1)$ but $a\notin C(\{a,b,c\},2)$. Now, let $S=\{a,c,d\}$ and $S'=\{b,c,d\}$. Note that $R(S,1)=R(S',1)=\{c,d\}$, $ C(S,2)\cap R(S,1)=\{c,d\}$, $C(S',2)\cap R(S',1)=\{c\}$, and therefore $C(S,2)\cap R(S,1)\neq C(S',2)\cap R(S',1)$, implying that $C$ violates \textit{the irrelevance of accepted alternatives}. 
\end{example}

The following example shows that \textit{CWARP} is necessary for the implication, that is, there exists a choice rule which satisfies \textit{capacity-filling} and \textit{monotonicity}, but not \textit{the irrelevance of accepted alternatives}.

\begin{example}
Let $A=\{a,b,c,d\}$. Let $\succ$ and $\succ'$ be defined as $a\succ b \succ c \succ d $ and $a\succ' b \succ' d \succ' c$. Let the choice rule $C$ be defined as follows. For each problem $(S, q)$ such that $a\in S$, $C$ chooses the highest $\succ$-priority alternatives from $S$ until $q$ alternatives are chosen or no alternative is left; and for each problem $(S, q)$ such that $a\notin S$, $C$ chooses the highest $\succ'$-priority alternatives from $S$ until $q$ alternatives are chosen or no alternative is left. Note that $C$ clearly satisfies \textit{capacity-filling} and \textit{monotonicity}. Now, let $S=\{a,c,d\}$ and $S'=\{b,c,d\}$. Note that $R(S,1)=R(S',1)=\{c,d\}$, $ C(S,2)\cap R(S,1)=\{c\}$, $C(S',2)\cap R(S',1)=\{d\}$, and therefore $C(S,2)\cap R(S,1)\neq C(S',2)\cap R(S',1)$, implying that $C$ violates \textit{the irrelevance of accepted alternatives}. 
\end{example}

\section{Independence of Properties in Theorem 1}\label{appendix:independence}

\noindent \textbf{Violating only capacity-filling:} Let $A=\{a,b,c\}$. Let $\succ$ be a priority ordering. Let $C$ be the choice rule such that, for each problem $(S,q)$, $C(S,q)$ is a singleton consisting of the $\succ$-maximal alternative in $S$. Note that $C$ violates \textit{capacity-filling} and clearly satisfies \textit{gross substitutes}. Since the choice does not vary with capacity, $C$ also satisfies \textit{monotonicity} and \textit{the irrelevance of accepted alternatives}.

\noindent \textbf{Violating only gross substitutes:} Let $A=\{a,b,c\}$. Let $\succ$ and $\succ'$ be defined as $a\succ b \succ c$ and $b\succ' a \succ' c$. Let the choice rule $C$ be defined as follows. For each problem $(S, q)$, $C(S,q)$ consists of the $\succ$-maximal alternative in $S$ if $q=1$ and $c\in S$; otherwise, $C(S,q)$ coincides with the choice rule that is responsive for $\succ'$. Note that $C$ satisfies \textit{capacity-filling}. 

Since $a\in C(\{a,b,c\},1)=\{a\}$ and $a\notin C(\{a,b\},1)=\{b\}$, $C$ violates \textit{gross substitutes}. 
To see that     $C$ satisfies \textit{monotonicity}, suppose that there exists a set $S$ and an alternative $x\in S$ such that $x\in C(S,1)$ and $x\notin C(S,2)$. Note that $x\notin C(S,2)$ implies that $x=c$ and $S=\{a,b,c\}$. But then, $x\notin C(S,1)=\{a\}$, a contradiction. To see that $C$ satisfies \textit{CWARP}, note that the revealed preference relation at $q=2$ consists of a unique pair: $b$ is revealed preferred to $c$. Then, by Lemma \ref{lemma_rejection_monotonicity}, $C$ also satisfies \textit{the irrelevance of accepted alternatives}.

\noindent \textbf{Violating only monotonicity:} Let $A=\{a,b,c\}$. Let $\succ$ be defined as $a\succ b \succ c$.  Let the choice rule $C$ be defined as follows. For each problem $(S,q)$, $C(S,q)$ consists of the $\succ$-maximal alternative in $S$ if $q=1$; $C(S,2)=S$ if $|S|=2$; and $C(\{a,b,c\},2)=\{b,c\}$. Note that $C$ satisfies \textit{capacity-filling}. 

Since $a\in C(\{a,b,c\},1)$ and $a\notin C(\{a,b,c\},2)$, $C$ violates monotonicity. For $q=1$, $C$ satisfies \textit{gross substitutes}, since $C$ maximizes $\succ$; for $q\in \{2,3\}$, $C$ clearly satisfies \textit{gross substitutes}. Since there are no two different problems with the same capacity and the same set of rejected alternatives, $C$ satisfies \textit{the irrelevance of accepted alternatives}.

\noindent \textbf{Violating only CWARP:} Note that three of the four rules that we discuss in Section~\ref{boston} satisfy all the properties but \textit{the irrelevance of accepted alternatives}.

\section{Importance of the Irrelevance of Satisfied Demand in Proposition~\ref{ek_corollary}}\label{appendix:example}

We provide an example of a mechanism which satisfies all the properties in the statement of Proposition~\ref{ek_corollary} except for \textit{the irrelevance of satisfied demand}, and therefore which is not a lexicographic deferred acceptance mechanism. The mechanism in the example is a deferred acceptance mechanism based on a choice structure such that the choice rule of each object is a walk-open choice rule. The example uses some arguments from the proof of Proposition~\ref{propp}, where it was shown that the walk-open choice rule violates \textit{CWARP}. 

\begin{example}
Let $N=\{a,b,c,d,e\}$ and let $O$ be a finite set of objects. Let $\succ^w$ be defined as $a\succ^w b \succ^w c \succ^w d \succ^w e$ and $\succ^o$ be defined as $e\succ^o b \succ^o d \succ^o c \succ^o a$. Let $(C_x)_{x\in O}$ be the choice structure such that for each object $x\in O$, $C_x$ is the walk-open choice rule based on $(\succ^w, \succ^o)$. Let $\varphi$ be the deferred acceptance mechanism based on the choice structure $(C_x)_{x\in O}$.  

Since for each $x\in O$, $C_x$ satisfies \textit{capacity-filling}, \textit{gross substitutes}, and \textit{monotonicity}, by Theorem $3$ of \cite{ehlers_klaus_2014b}, $\varphi$ satisfies \textit{unavailable-type-invariance}, \textit{weak non-wastefulness}, \textit{resource-monotonicity}, \textit{truncation-invariance}, and \textit{strategy-proofness}.

Let $x\in O$. Let $q$ be such that $q_x=2$ and for each $y\in O\setminus \{x\}$, $q_y=0$. Let $R$ be such that $x$ is preferred to $\emptyset$ for all the agents except for $b$. Note that $D_x(\varphi(R,q),R)=\{c,d\}$ since $C_x(\{a,c,d,e\},2)=\{a,e\}$. Let $R'$ be such that $x$ is preferred to $\emptyset$ for all the agents except for $e$. Note that $D_x(\varphi(R',q),R')=\{c,d\}$ since $C_x(\{a,b,c,d\},2)=\{a,b\}$. Thus, $D_x(\varphi(R,q),R)=D_x(\varphi(R',q),R')$.

Now, we have $D_x(\varphi(R,q+1_x),R)=\{d\}$ since $C_x(\{a,c,d,e,3\})=\{a,c,e\}$ and $D_x(\varphi(R',q+1_x),R')=\{c\}$ since $C_x(\{a,b,c,d,3\})=\{a,b,d\}$. (Note that when $q=3$, $C_x$ is lexicographic for $(\succ_w,\succ_w,\succ_o)$.) Hence, $\varphi$ violates \textit{the irrelevance of satisfied demand}.
\end{example}

\section{Responsive Choice}
\label{responsive}

A well-known example of a lexicographic choice rule is a ``responsive'' choice rule,\footnote{Responsive choice rules have been studied particularly in the two-sided matching context \citep{roth_sotomayor_1990}.} which is lexicographic for a priority profile where all the priority orderings are the same. Formally, a choice rule $C$ is \textbf{responsive for a priority ordering} $\succ$ if for each $(S,q)\in \mathcal{A} \times \{1,\ldots ,n\}$, $C(S,q)$ is obtained by choosing the highest $\succ$-priority alternatives in $S$ until $q$ alternatives are chosen or no alternative is left. Note that $C$ is responsive for $\succ$ if and only if it is lexicographic for the priority profile $(\succ,\ldots ,\succ)$.

\cite{chambers_yenmez_2018} characterize ``responsive'' choice rules, but in the context of ``classical'' choice problems which do not explicitly refer to a variable capacity parameter. Formally, a classical choice rule is a function $C:\mathcal{A}\rightarrow \mathcal{A}$ such that for each $S\in \mathcal{A}$, $C(S)\subseteq S$. A classical choice rule is \textit{responsive} if there exists a priority ordering $\succ$ and a capacity $q\in \{1,\ldots ,n\}$ such that for each $S\in \mathcal{A}$, $C(S)$ is obtained by choosing the highest $\succ$-priority alternatives until the capacity $q$ is reached or no alternative is left. \cite{chambers_yenmez_2018} show that a classical choice rule satisfies \textit{capacity-filling}\footnote{A classical choice rule satisfies \textit{capacity-filling} if there exists a capacity such that at each choice problem, an alternative is rejected only if the capacity is reached.} and the \textit{weaker axiom of revealed preference (WrARP)} if and only if it is responsive.\footnote{\cite{chambers_yenmez_2018} also provide a characterization of choice rules that are responsive for a known capacity (namely $q$-responsive choice rules).} \textit{WrARP}\footnote{\textit{WrARP} was introduced by \cite{jamison} and also studied by \cite{es2008}.} requires that for each pair $a,b\in A$ and $S,S'\in \mathcal{A}$ such that $a,b\in S\cap S'$ , $$\mbox{if }a\in C(S) \mbox{ and }b\in C(S')\setminus C(S),\mbox{ then }a\in C(S').$$   

To see what \cite{chambers_yenmez_2018} implies in our variable capacity setup, consider the following extension of \textit{WrARP} to our setup.

\noindent \textbf{Weaker axiom of revealed preference (WrARP):} For each $S,S'\in \mathcal{A}$, $q \in \{1,\ldots ,n\}$, and each pair $a,b\in S\cap S'$, $$\mbox{if }a\in C(S,q) \mbox{ and }b\in C(S',q)\setminus C(S,q),\mbox{ then }a\in C(S',q).$$

The following Proposition~\ref{chambers_yenmez} directly follows from \cite{chambers_yenmez_2018}. 

\begin{proposition}
\label{chambers_yenmez}
A choice rule satisfies \textit{capacity-filling} and the \textit{weaker axiom of revealed preference} if and only if for each $q\in \{1,\ldots ,n\}$, there is a priority ordering $\succ^q$ such that for each $S\in \mathcal{A}$, $C(S,q)$ is obtained by choosing the highest $\succ^q$-priority alternatives until the capacity $q$ is reached or no alternative is left. 
\end{proposition}

Proposition~\ref{chambers_yenmez} states that \textit{capacity-filling} and\textit{WrARP} characterizes ``capacity-wise responsive'' choice rules, which are responsive for each capacity, but the associated priority orderings for different capacities may be different. Yet, a characterization of responsive choice rules in our setup does not directly follow from \cite{chambers_yenmez_2018}.   

We show that, the following extension of \textit{WrARP}, together with \textit{capacity-filling}, characterizes responsive choice rules in our variable-capacity setup. The property, called the \textit{capacity-wise weaker axiom of revealed preference (CWrARP)}, requires that if an alternative $a$ is chosen and $b$ is not chosen at a problem where they are both available, then at any problem where they are both available, $a$ is chosen whenever $b$ is chosen.

\noindent \textbf{Capacity-wise weaker axiom of revealed preference (CWrARP):} For each $S,S'\in \mathcal{A}$, $q,q' \in \{1,\ldots ,n\}$, and each pair $a,b\in S\cap S'$ , $$\mbox{if }a\in C(S,q) \mbox{ and }b\in C(S',q')\setminus C(S,q),\mbox{ then }a\in C(S',q').$$

\begin{theorem}
A choice rule is responsive if and only if it satisfies \textit{capacity-filling} and the \textit{capacity-wise weaker axiom of revealed preference}. 
\end{theorem}

\begin{proof}
It is clear that a responsive choice rule satisfies \textit{capacity-filling} and \textit{CWrARP}. Let $C$ be a choice rule satisfying \textit{capacity-filling} and \textit{CWrARP}. Clearly, \textit{CWrARP} implies \textit{WrARP}, and therefore by Proposition~\ref{chambers_yenmez}, for each $q\in \{1,\ldots ,n\}$, there is a priority ordering $\succ^q$ such that for each $S\in \mathcal{A}$, $C(S,q)$ is obtained by choosing the highest $\succ^q$-priority alternatives until the capacity $q$ is reached or no alternative is left. 

Let $(S,q)\in \mathcal{A}\times \{1,\ldots ,n\}$. If $|S|\leq q$, then by \textit{capacity-filling}, $C(S,q)=S$. Suppose that $|S|>q$. First note that $C(S,q-1)\subseteq C(S,q)$, since otherwise, by \textit{capacity-filling}, there is a pair $a,b\in S$ such that $a\in C(S,q-1)\setminus C(S,q)$ and $b\in C(S,q)\setminus C(S,q-1)$, which contradicts \textit{CWrARP}. Now, consider any pair $a,b\in R(S,q-1)$ such that $a\in C(S,q)$ and $b\notin C(S,q)$. By \textit{CWrARP}, for any $S'\in \mathcal{A}$, $b$ is not chosen over $a$ at $(S',q)$, implying that $a$ has $\succ^q$-priority over $b$. But then, for each $S\in \mathcal{A}$, $C(S,q)$ is obtained by choosing the highest $\succ^{q-1}$-priority alternatives until the capacity $q$ is reached or no alternative is left. Since we started with an arbitrary $q\in \{1,\ldots ,n\}$, $C$ is a choice rule that is responsive to $\succ^1$.
\end{proof}

\end{appendix}

\end{document}